\documentclass[conference,a4paper]{IEEEtran}
\usepackage{graphicx,color}   %color for pdf_t generated by xfig
\usepackage[cmex10]{amsmath}  %cmex10 ensures that only type 1 fonts will be used
\usepackage{amssymb}
\usepackage{cite}  %combined consecutive references
\usepackage[caption=false,font=footnotesize]{subfig} %IEEE compliance
\usepackage[utf8]{inputenc}
\usepackage[T1]{fontenc}
\usepackage{verbatim}
\newtheorem{theorem}{Theorem}
\newtheorem{remark}{Remark}
\newtheorem{lemma}{Lemma}
\newtheorem{definition}{Definition}

\long\def\symbolfootnote[#1]#2{\begingroup%
\def\thefootnote{\fnsymbol{footnote}}\footnote[#1]{#2}\endgroup}

\begin{document}

%---------- Title ----------
\title{Functional-Decode-Forward for the General Discrete Memoryless Two-Way Relay Channel}
\author{\IEEEauthorblockN{Lawrence Ong, Christopher M. Kellett, and Sarah J. Johnson}
\IEEEauthorblockA{School of Electrical Engineering and Computer Science, 
The University of Newcastle\\
Email: lawrence.ong@cantab.net; \{chris.kellett, sarah.johnson\}@newcastle.edu.au}
}
\maketitle

\begin{abstract}
We consider the general discrete memoryless two-way relay channel, where two users exchange messages via a relay, and propose two functional-decode-forward coding strategies for this channel. Functional-decode-forward involves the relay decoding a function of the users' messages rather than the individual messages themselves. This function is then broadcast back to the users, which can be used in conjunction with the user's own message to decode the other user's message. Via a numerical example, we show that functional-decode-forward with linear codes is capable of achieving strictly larger sum rates than those achievable by other strategies.
\end{abstract}

\IEEEpeerreviewmaketitle

\section{Introduction}\label{sec:introduction}\symbolfootnote[0]{This work is supported by the Australian Research Council under grants DP087725 and DP1093114.}
We obtain two new achievable rate regions for the general discrete memoryless two-way relay channel (TWRC), in which two users exchange messages through a relay. We consider TWRCs with no direct link between the users (see Fig.~\ref{fig:twrc}). The new rate regions are obtained using the idea of \emph{functional-decode-forward} (FDF), where the relay only decodes a function of the users' messages or codewords without needing to decode the messages or codewords themselves (hence saving the \emph{uplink} bandwidth from the users to the relay). The relay then broadcasts the function to both users. The function must be defined such that knowing its own message, each user is able to decode the message sent by the other user. 

We first illustrate the concept of FDF using the \emph{noiseless} binary adder TWRC as an example, where nodes 1 and 2 (the users) exchange data through node 3 (the relay). Let $X_i \in \{0,1\}$ be node $i$'s transmitted signal and $Y_i \in \{0,1\}$ be node $i$'s received signal. The noiseless binary adder TWRC is defined as follows: (i) the uplink is $Y_3 = X_1 + X_2 \mod 2$, and (ii) the downlink is $Y_1 = X_3$ and $Y_2 = X_3$.
Assume that the source messages are in bits, i.e. $W_1, W_2 \in \{0,1\}$. The well-known optimal (rate-maximizing) coding strategy is for the users to transmit uncoded information bits, i.e., $X_i=W_i$, for  $i \in \{1,2\}$, and for the relay to forward its received bits, i.e., $Y_3=X_3$. Having received $Y_1$ which is $W_1 + W_2 \mod 2$, and knowing its own message $W_1$, node 1 can recover $W_2$ perfectly. Node 2 can recover $W_1$ similarly. Here, the capacity of 1 bit/channel use is achievable using this strategy.

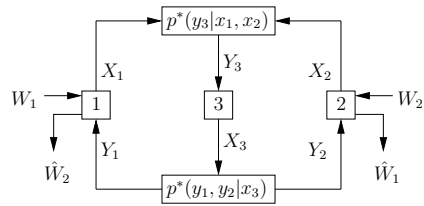
\begin{figure}[t]
\centering
\resizebox{5.2cm}{!}{
\begin{picture}(0,0)%
\includegraphics{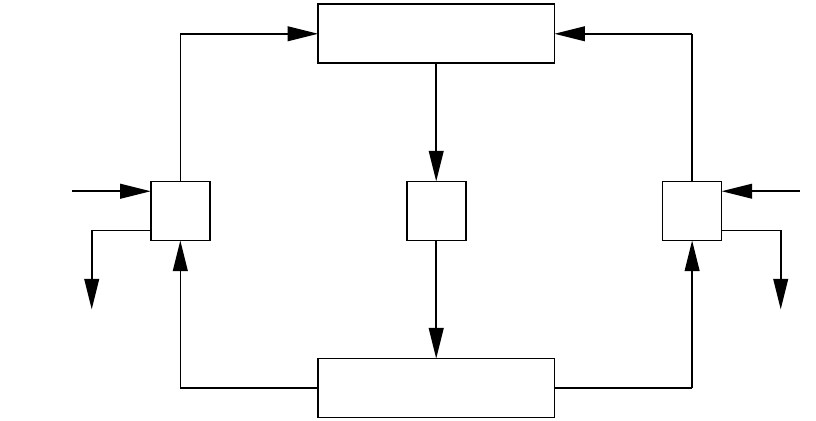}%
\end{picture}%
\setlength{\unitlength}{4144sp}%
\begingroup\makeatletter\ifx\SetFigFont\undefined%
\gdef\SetFigFont#1#2#3#4#5{%
  \fontsize{#1}{#2pt}%
  \fontfamily{#3}\fontseries{#4}\fontshape{#5}%
  \selectfont}%
\fi\endgroup%
\begin{picture}(3720,1914)(391,-1873)
\put(3511,-961){\makebox(0,0)[lb]{\smash{{\SetFigFont{12}{14.4}{\familydefault}{\mddefault}{\updefault}{\color[rgb]{0,0,0}$2$}%
}}}}
\put(2341,-961){\makebox(0,0)[lb]{\smash{{\SetFigFont{12}{14.4}{\familydefault}{\mddefault}{\updefault}{\color[rgb]{0,0,0}$3$}%
}}}}
\put(1171,-961){\makebox(0,0)[lb]{\smash{{\SetFigFont{12}{14.4}{\familydefault}{\mddefault}{\updefault}{\color[rgb]{0,0,0}$1$}%
}}}}
\put(1891,-151){\makebox(0,0)[lb]{\smash{{\SetFigFont{12}{14.4}{\familydefault}{\mddefault}{\updefault}{\color[rgb]{0,0,0}$p^*(y_3|x_1,x_2)$}%
}}}}
\put(1891,-1771){\makebox(0,0)[lb]{\smash{{\SetFigFont{12}{14.4}{\familydefault}{\mddefault}{\updefault}{\color[rgb]{0,0,0}$p^*(y_1,y_2|x_3)$}%
}}}}
\put(721,-1636){\makebox(0,0)[lb]{\smash{{\SetFigFont{12}{14.4}{\familydefault}{\mddefault}{\updefault}{\color[rgb]{0,0,0}$\hat{W}_2$}%
}}}}
\put(3871,-1636){\makebox(0,0)[lb]{\smash{{\SetFigFont{12}{14.4}{\familydefault}{\mddefault}{\updefault}{\color[rgb]{0,0,0}$\hat{W}_1$}%
}}}}
\put(2431,-556){\makebox(0,0)[lb]{\smash{{\SetFigFont{12}{14.4}{\familydefault}{\mddefault}{\updefault}{\color[rgb]{0,0,0}$Y_3$}%
}}}}
\put(2431,-1321){\makebox(0,0)[lb]{\smash{{\SetFigFont{12}{14.4}{\familydefault}{\mddefault}{\updefault}{\color[rgb]{0,0,0}$X_3$}%
}}}}
\put(1261,-646){\makebox(0,0)[lb]{\smash{{\SetFigFont{12}{14.4}{\familydefault}{\mddefault}{\updefault}{\color[rgb]{0,0,0}$X_1$}%
}}}}
\put(3241,-646){\makebox(0,0)[lb]{\smash{{\SetFigFont{12}{14.4}{\familydefault}{\mddefault}{\updefault}{\color[rgb]{0,0,0}$X_2$}%
}}}}
\put(3241,-1411){\makebox(0,0)[lb]{\smash{{\SetFigFont{12}{14.4}{\familydefault}{\mddefault}{\updefault}{\color[rgb]{0,0,0}$Y_2$}%
}}}}
\put(1261,-1411){\makebox(0,0)[lb]{\smash{{\SetFigFont{12}{14.4}{\familydefault}{\mddefault}{\updefault}{\color[rgb]{0,0,0}$Y_1$}%
}}}}
\put(4096,-916){\makebox(0,0)[lb]{\smash{{\SetFigFont{12}{14.4}{\familydefault}{\mddefault}{\updefault}{\color[rgb]{0,0,0}$W_2$}%
}}}}
\put(406,-916){\makebox(0,0)[lb]{\smash{{\SetFigFont{12}{14.4}{\familydefault}{\mddefault}{\updefault}{\color[rgb]{0,0,0}$W_1$}%
}}}}
\end{picture}%

}
\caption{The general discrete memoryless TWRC considered in this paper}
\label{fig:twrc}
\end{figure}

While the bit-wise modulo-two addition of the users' messages seems to be a good function for the relay to transmit, the main challenge of FDF on a \emph{noisy} TWRC lies in:
\begin{itemize}
\item selecting a \emph{good function} of the users' messages/codewords which the relay should decode, and
\item constructing \emph{good codes} for the users such that the relay can efficiently decode this function without needing to decode the individual users' messages/codewords.
\end{itemize}

%\subsection{Functional-Decode-Forward in Noisy Additive Channels}
In the case of \emph{adder channels}, e.g.,  $Y_3 = X_1 + X_2 + N$, where $N$ is the channel noise, linear codes can be used (see \cite{ongjohnsonkellett10cl}  for the case of binary adder channels, \cite{ong10iccs} for finite field adder channels, and \cite{narayananwilson07,namchung08} for AWGN channels). Let $\boldsymbol{X}_i$ be user $i$'s length-$n$ linear codeword\footnote{Bold letters are used to denote a block of $n$ channel uses, e.g. $\boldsymbol{X}_i = (X_i[1], X_i[2], \dotsc, X_i[n])$, where $X_i[t]$ is $X_i$ on the $t$-th channel use.}, for $i \in \{1,2\}$. The structure of linear codes guarantees that $\boldsymbol{U} \triangleq (\boldsymbol{X}_1 + \boldsymbol{X}_2)$ is a codeword from the same code. The relay effectively receives $\boldsymbol{Y}_3 = \boldsymbol{U} + \boldsymbol{N}$, which is a noisy version of $\boldsymbol{U}$.  Capacity-achieving linear codes have been shown to exist for this type of additive noise channel. This means if the users transmit using these linear codes, then the relay is able to efficiently decode $\boldsymbol{U}$ (which is a function of the users' codewords) without having to decode the users' codewords individually. The relay then broadcasts $\boldsymbol{U}$ to the users, and each user can obtain the other user's message from $\boldsymbol{U}$ and its own message/codeword.

For the above adder channels, the channels actually perform the desired function by adding the users' codewords. For FDF on \emph{general} discrete memoryless TWRCs in which the channels do not ``help'', it is not immediately obvious what function the relay should decode, and how the relay can decode the function without first decoding the individual messages.

In this paper, we use random linear codes for FDF on the general discrete memoryless TWRC following the idea in \cite{nazergastpar08allerton} for the multiple-access channel, i.e., the users transmit randomly generated linear codewords on the uplink. Although the uplink output $Y_3$ cannot be written as a (noisy) function of $X_1 + X_2$, by invoking the Markov Lemma we will prove that the relay is still able to \emph{reliably} (i.e., with arbitrarily small error probability) decode $\boldsymbol{X}_1 + \boldsymbol{X}_2$ without needing to decode the individual messages/codewords. The relay then broadcasts $\boldsymbol{X}_1 + \boldsymbol{X}_2$ to the users for each of them to obtain the other user's message. We call this strategy functional-decode-forward with linear codes (FDF-L).

Another method for the relay to decode a function of the users' messages in the general discrete memoryless TWRC is by using \emph{systematic computation codes}~\cite{nazergastpar07} on the uplink. On the uplink, the users first send uncoded data, followed by linear-coded signals. After the relay decodes a function of the users' messages, the downlink transmission is the same as that in FDF-L. We call this strategy functional-decode-forward with systematic computation codes (FDF-S).

We will first derive two achievable rate regions for the general discrete memoryless TWRC, using FDF-L and FDF-S. We will then show, using an example, that FDF-L can achieve higher sum rates than those achievable by FDF-S and by existing coding strategies for the TWRC, including (i) the \emph{complete-decode-forward} (CDF) coding strategy\footnote{The strategy is commonly known as decode-and-forward or decode-forward. We modified the name of this strategy here to reflect that the relay completely decodes both the users' messages before forwarding them.}, where the relay fully decodes the messages from both users, re-encodes and broadcasts a function of the messages back to the users~\cite{knopp06,kimmitrantarokh08}, and (ii) the \emph{compress-forward} (CF) coding strategy, where the relay quantizes its received signals, re-encodes and broadcasts the quantized signals to users~\cite{schnurroechtering07}. 

%The rest of this paper is organized as follows. We define the channel model in Sec.~\ref{sec:model}, and random linear codes in Sec.~\ref{sec:fields}. We derive an achievable rate region for the discrete memoryless TWRC using FDF-L in Sec.~\ref{sec:rate_region}, and compare the achievable rate regions of different coding strategies in Sec.~\ref{sec:comparison}.

\section{Channel Model} \label{sec:model}

%\subsection{A Two-Way Relay Channel (TWRC)}

Fig.~\ref{fig:twrc} depicts the general discrete memoryless TWRC considered in the paper, where users 1 and 2 exchange data through the relay (node 3). We denote by $X_i\in \mathcal{X}_i$ the channel input from node $i$, $Y_i\in\mathcal{Y}_i$ the channel output received by node $i$, and $W_i\in\mathcal{W}_i$ user $i$'s message. The TWRC can be completely defined by (i) the uplink channel $p^*(y_3|x_1,x_2)$, and (ii) the downlink channel $p^*(y_1,y_2|x_3)$.

Let $W_i \in \left\{1,2,\dotsc,2^{nR_i}\right\}$ be an $(nR_i)$-bit message, for $i \in \{1,2\}$. Consider on each uplink and downlink, $n$ channel uses. User $i$ transmits $\boldsymbol{X}_i(W_i) = f_i(W_i)$, for $i \in \{1,2\}$. At any time, the relay transmits a function of its previously received signals, i.e., $X_3[t] = f_{3,t}(Y_3[1],Y_3[2], \dots, Y_3[t-1])$, for $t\in \{1,2,\dotsc,n\}$. After $n$ channel uses, each user estimates the message of the other user from its received signals and its own message, i.e., $\hat{W}_2 = g_1(\boldsymbol{Y}_1,W_1)$ and $\hat{W}_1 = g_2(\boldsymbol{Y}_2, W_2)$ for users 1 and 2 respectively. $\hat{W}_2$ is node 1's estimate of $W_2$, and vice versa. Assuming that the message pair $(W_1,W_2)$ is uniformly distributed in $\mathcal{W}_1 \times \mathcal{W}_2$, a rate pair $(R_1,R_2)$ is said to be \emph{achievable} if each user can reliably decode the messages of the other. We say that a user can reliably decode a message if the probability that it wrongly decodes the message can be made arbitrarily small.

\section{Fields and Linear Codes} \label{sec:fields}

Now, we will present a construction of \emph{random}
linear codes with elements from finite fields. Random linear codes will be used for the users to transmit their respective messages to the relay.  Using random linear codes, any two codewords are statistically pair-wise independent. This property is important for proving reliable communications. Furthermore, as mentioned in Sec.~\ref{sec:introduction}, the structure of linear codes enables the relay to decode the desired function of the user's codewords without needing to decode the individual codewords or  messages.

Let $\mathcal{F}$ be a finite field with associated operations of addition $\oplus$ and multiplication $\odot$. Consider the following codeword generating function that maps a message $\boldsymbol{s} \in \mathcal{F}^k$ to a codeword $\boldsymbol{x} \in \mathcal{F}^n$:
\vspace{-1.5ex}
\begin{equation}
\boldsymbol{x} = (\boldsymbol{s} \odot \mathbb{G})  \oplus \boldsymbol{q}, \label{eq:linear-codes-def-1}
\end{equation}
% \begin{subequations}
% \begin{align}
% \boldsymbol{x} &= (\boldsymbol{s} \odot \mathbb{G})  \oplus \boldsymbol{q}\label{eq:linear-codes-def-1}\\
% &= \left( \boldsymbol{s} \odot \begin{bmatrix} \boldsymbol{g}_1 \\ \vdots \\ \boldsymbol{g}_k \end{bmatrix} \right) \oplus \boldsymbol{q},\label{eq:linear-codes-def-2}
% \end{align}
% \end{subequations}
where $\boldsymbol{x}$ is a row vector of length $n$, $\boldsymbol{s}$ is a row vector of length $k$, $\mathbb{G}$ is a fixed $k$-by-$n$ matrix, with each element independently and uniformly chosen over $\mathcal{F}$, and $\boldsymbol{q}$ is a fixed row vector of length $n$, with each element independently and uniformly chosen over $\mathcal{F}$.
We extend Gallager's results for binary linear codes \cite[p. 207]{gallager68} to finite-field linear codes in the following two lemmas (see \cite{ong10iccs} for the proofs):

\begin{lemma}\label{lemma:linear-codes-1}
Consider the linear codes defined in \eqref{eq:linear-codes-def-1}. Over the ensemble of codes, the probability that a message $\boldsymbol{s}_1$ is mapped to a given codeword $\boldsymbol{x}_1$ is $p(\boldsymbol{x}_1)=|\mathcal{F}|^{-n}$.
\end{lemma}

\begin{lemma}\label{lemma:linear-codes-2}
Consider the linear codes defined in \eqref{eq:linear-codes-def-1}. Let $\boldsymbol{s}_1$ and $\boldsymbol{s}_2$ be any two different messages. The corresponding codewords $\boldsymbol{x}_1 = (\boldsymbol{s}_1 \odot \mathbb{G}) \oplus \boldsymbol{q}$ and $\boldsymbol{x}_2 = (\boldsymbol{s}_2 \odot \mathbb{G}) \oplus \boldsymbol{q}$ are independent.
\end{lemma}

Besides the above extensions of Gallager's results for binary linear codes, we have the following additional result.
\begin{lemma}\label{lemma:two-linear-codes}
Consider two linear codes: $\boldsymbol{x} = \boldsymbol{s}_1 \odot \mathbb{G} \oplus \boldsymbol{q}_1$, and $\boldsymbol{v} = \boldsymbol{s}_2 \odot \mathbb{G} \oplus \boldsymbol{q}_2$, where $\mathbb{G}$, $\boldsymbol{q}_1$, and $\boldsymbol{q}_2$ are independently generated according to the uniform distribution. Any two codewords, one from each code, are independent.
\end{lemma}

\begin{proof}
From Lemma~\ref{lemma:linear-codes-1}, we know that $p(\boldsymbol{x})=|\mathcal{F}|^{-n}$ and $p(\boldsymbol{v})=|\mathcal{F}|^{-n}$. We have to show that $p(\boldsymbol{x},\boldsymbol{v})=|\mathcal{F}|^{-2n}$. Elements in $\mathbb{G}$, $\boldsymbol{q}_1$, and $\boldsymbol{q}_2$ are independent and uniformly distributed, and so each $(\mathbb{G},\boldsymbol{q}_1,\boldsymbol{q}_2)$ has a probability of $|\mathcal{F}|^{-n(k+2)}$ of being selected. For any given $\mathbb{G}$, there is only one $\boldsymbol{q}_1$ and one $\boldsymbol{q}_2$ that results in the given $\boldsymbol{x}$ and $\boldsymbol{v}$. So, there are only $|\mathcal{F}|^{-nk}$ different $(\mathbb{G},\boldsymbol{q}_1,\boldsymbol{q}_2)$ that map $\boldsymbol{s}_1$ to $\boldsymbol{x}$, and $\boldsymbol{s}_2$ to $\boldsymbol{v}$. So, $p(\boldsymbol{x},\boldsymbol{v})= |\mathcal{F}|^{-nk}|\mathcal{F}|^{-n(k+2)}=  |\mathcal{F}|^{-2n}$.
\end{proof}

\begin{remark}
The use of dither $\boldsymbol{q}$ in the above-defined linear codes of the form \eqref{eq:linear-codes-def-1} is essential in proving Lemmas~\ref{lemma:linear-codes-1} and \ref{lemma:two-linear-codes}.
\end{remark}

\section{Two New Achievable Rate Regions} \label{sec:rate_region}

\subsection{Functional-Decode-Forward with Linear Codes (FDF-L)}

We first prove the following achievable rate region for the discrete memoryless TWRC using FDF-L:

\begin{theorem}\label{theorem:achievability}
Consider a TWRC where $|\mathcal{X}_1| = |\mathcal{X}_2| = |\mathcal{F}|$, for some finite field $\mathcal{F}$. Rename\footnote{We choose a renaming scheme that maximizes $I(U;Y_3)$.}  the elements in $\mathcal{X}_1$ and $\mathcal{X}_2$ so that $\mathcal{X}_1 = \mathcal{X}_2 = \mathcal{F}$. The rate pair $(R_1,R_2)$ is achievable if
\vspace{-1.5ex}
\begin{equation}
R_1,R_2 \leq \min \{ I(U;Y_3), I(X_3;Y_1), I(X_3;Y_2) \},\label{eq:fdf-l}
\end{equation}
for
\vspace{-0.6ex}
\begin{equation}
p(u,y_3) = \frac{1}{|\mathcal{F}|^{2}} \sum_{\substack{x_1,x_2\\\text{s.t. }x_1\oplus x_2=u }} p^*(y_3|x_1,x_2), \label{eq:u-y3}
\end{equation}
where $U \in \mathcal{F}$, and for $p(x_3,y_1,y_2) = p(x_3)p^*(y_1,y_2|x_3)$.
% \vspace{-1ex}
% \begin{equation}
% p(x_3,y_1,y_2) = p(x_3) p^*(y_1,y_2|x_3),
% \end{equation}
\end{theorem}

\begin{remark}\label{remark:non-finite-field}
Recall that $\mathcal{F}$ is a finite field iff $|\mathcal{F}| = \ell ^z$ for some $\ell \in \mathcal{Z}_\text{P}$ (prime numbers) and some $z \in \mathcal{Z}_+$ (positive integers). For TWRCs where $|\mathcal{X}_1| \neq |\mathcal{X}_2|$ or $|\mathcal{X}_1| = |\mathcal{X}_2| \neq \ell^z$, $\forall \ell \in \mathcal{Z}_\text{P}$, $\forall z \in \mathcal{Z}_+$, we select subsets $\mathcal{X}_1' \subseteq \mathcal{X}_1$ and $\mathcal{X}_2' \subseteq \mathcal{X}_2$ such that $|\mathcal{X}_1'| = |\mathcal{X}_2'| = \ell^z$, for some $\ell \in \mathcal{Z}_\text{P}$, $z \in \mathcal{Z}_+$. The result in Theorem~\ref{theorem:achievability} holds for any discrete memoryless TWRC with $X_i$ replaced by $X_i' \in \mathcal{X}_i'$, for $i \in \{1,2\}$.
\end{remark}

\subsubsection{The Auxiliary Random Variable $U$}

The auxiliary random variable $U \in \mathcal{F}$ is the information that the relay recovers from its received signal $Y_3$, and broadcasts to both users 1 and 2. Before proceeding to the proof of Theorem~\ref{theorem:achievability}, we derive $p(u,y_3)$ in \eqref{eq:u-y3}. Define
\vspace{-0.2ex}
\begin{equation}
U \triangleq X_1 \oplus X_2. \label{definition-of-u}
\end{equation}
%where $\oplus$ is chosen to maximize $I(U;Y_3)$.
This means $p(u|x_1,x_2)=\mathbf{1}\big(u=x_1 \oplus x_2\big)$, where $\mathbf{1}(E)$ is 1 if the event $E$ is true, and is 0 otherwise. We can write
\vspace{-0.2ex}
\begin{equation}
p(u,x_1,x_2,y_3) = p(x_1,x_2)p^*(y_3|x_1,x_2) \mathbf{1}\big(u=x_1 \oplus x_2\big), \nonumber
\end{equation}
meaning $U - (X_1,X_2) - Y_3$ forms a Markov chain.

\begin{remark}
The function that the relay should decode, $U$, is not unique. 
Other functions are possible as long as each user can obtain the other user's message from the function and its own message.
\end{remark}

In FDF-L, we use the linear code structure
 in \eqref{eq:linear-codes-def-1} with the same $\mathbb{G}$, and independent $\boldsymbol{q}_1$ and $\boldsymbol{q}_2$ for users 1 and 2 respectively (see \eqref{eq:user-code-construction} for code construction). So, for any pair of source messages, from Lemma~\ref{lemma:two-linear-codes}, we have $p(x_1,x_2) = \frac{1}{|\mathcal{F}|^2}$, i.e., $X_1$ and $X_2$ are independent of each other and are uniformly distributed. Hence,
\vspace{-0.2ex}
\begin{subequations}
\begin{align}
p(u,y_3) &= \sum_{x_1,x_2} p(u,x_1,x_2,y_3)\\
&= \sum_{x_1,x_2} \frac{1}{|\mathcal{F}|^2} \mathbf{1}\big(u=x_1 \oplus x_2 \big) p^*(y_3|x_1,x_2)\\
&= \frac{1}{|\mathcal{F}|^2} \sum_{\substack{x_1,x_2\\\text{s.t. }x_1\oplus x_2=u }} p^*(y_3|x_1,x_2). \label{eq:p-u-y3}
\end{align}
\end{subequations}

\subsubsection{Proof of Theorem~\ref{theorem:achievability}}
Fig.~\ref{fig:coding} depicts the relationship among the random variables for FDF-L used to achieve the rates in Theorem~\ref{theorem:achievability}.

\noindent{\bf Uplink:}\\
\indent Assuming that $|\mathcal{W}_1| = |\mathcal{W}_2| = 2^{nR}$, on the uplink, for a sufficiently large $n$, we choose $k$ such that $|\mathcal{F}|^k = 2^{nR}$, and define a bijective mapping from $W_i \in \{1,2,\dotsc,2^{nR}\}$ to $\boldsymbol{S}_i \in \mathcal{F}^k$, for $i \in \{1,2\}$. The users transmit using linear codes of the form \eqref{eq:linear-codes-def-1}, i.e.,
\vspace{-0.4ex}
\begin{equation}
\boldsymbol{X}_i (W_i) = ( \boldsymbol{S}_i \odot \mathbb{G} ) \oplus \boldsymbol{q}_i, \label{eq:user-code-construction}
\end{equation}
for $i \in \{1,2\}$, where $\boldsymbol{S}_i \in \mathcal{F}^k$, $\mathbb{G} \in \mathcal{F}^{k \times n}$, and $\boldsymbol{X}_i, \boldsymbol{q}_i \in \mathcal{F}^n$. All elements in $\mathbb{G}$, $\boldsymbol{q}_1$, and $\boldsymbol{q}_2$  are uniformly and independently chosen over $\mathcal{F}$, and are fixed for all channel uses. 

From definition \eqref{definition-of-u} we have $\boldsymbol{U} \triangleq \boldsymbol{X}_1 \oplus \boldsymbol{X}_2 = (\boldsymbol{S}_{3}  \odot \mathbb{G})  \oplus \boldsymbol{q}_3$, which is also a codeword from a linear code of the form \eqref{eq:linear-codes-def-1}, where $\boldsymbol{S}_{3} = \boldsymbol{S}_1 \oplus \boldsymbol{S}_2$ is the ``message'' and $\boldsymbol{q}_3= \boldsymbol{q}_1 \oplus \boldsymbol{q}_2$. The relay estimates the codeword $\boldsymbol{U}(V_3)$ from its received signals $\boldsymbol{Y}_3$, where $V_3 \in \{1,\dotsc,2^{nR}\}$ is mapped (bijectively) to $\boldsymbol{S}_3 \in \mathcal{F}^k$ the same way as $W_i$ is to $\boldsymbol{S}_i$, for $i \in \{1,2\}$.

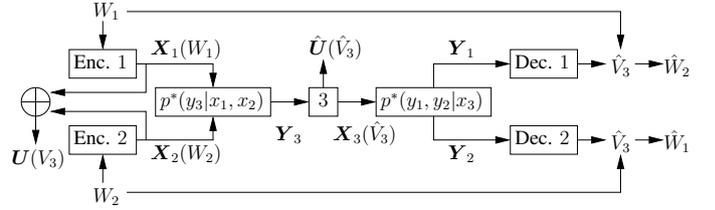
\begin{figure}[t]
%\centering
\hspace{-2mm}
\resizebox{8.7cm}{!}{
\begin{picture}(0,0)%
\includegraphics{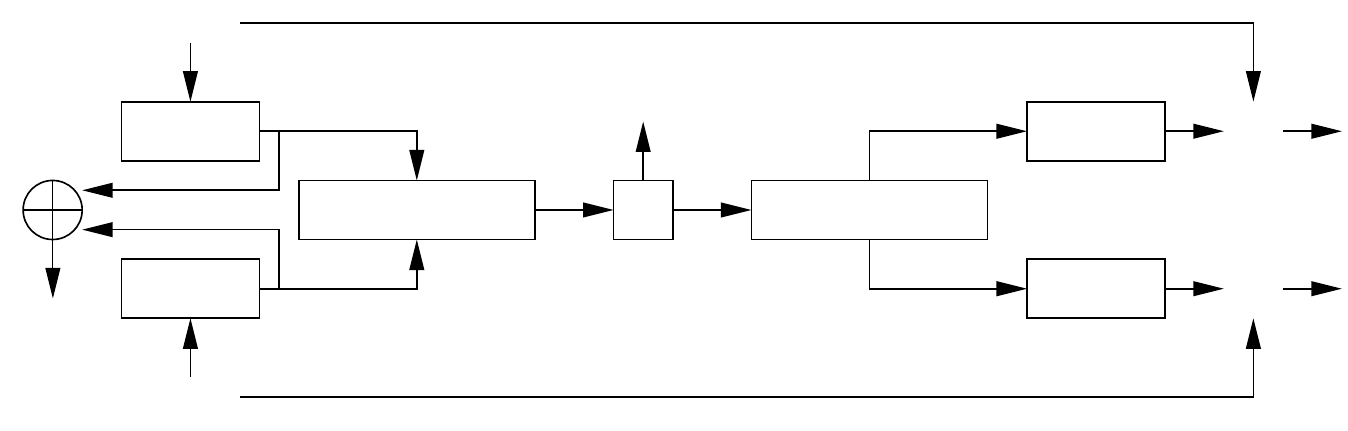}%
\end{picture}%
\setlength{\unitlength}{4144sp}%
\begingroup\makeatletter\ifx\SetFigFont\undefined%
\gdef\SetFigFont#1#2#3#4#5{%
  \reset@font\fontsize{#1}{#2pt}%
  \fontfamily{#3}\fontseries{#4}\fontshape{#5}%
  \selectfont}%
\fi\endgroup%
\begin{picture}(6150,1966)(-284,-4175)
\put(1126,-3211){\makebox(0,0)[lb]{\smash{{\SetFigFont{12}{14.4}{\familydefault}{\mddefault}{\updefault}{\color[rgb]{0,0,0}$p^*(y_3|x_1,x_2)$}%
}}}}
\put(2611,-3211){\makebox(0,0)[lb]{\smash{{\SetFigFont{12}{14.4}{\familydefault}{\mddefault}{\updefault}{\color[rgb]{0,0,0}$3$}%
}}}}
\put(3196,-3211){\makebox(0,0)[lb]{\smash{{\SetFigFont{12}{14.4}{\familydefault}{\mddefault}{\updefault}{\color[rgb]{0,0,0}$p^*(y_1,y_2|x_3)$}%
}}}}
\put(316,-2851){\makebox(0,0)[lb]{\smash{{\SetFigFont{12}{14.4}{\familydefault}{\mddefault}{\updefault}{\color[rgb]{0,0,0}Enc. $1$}%
}}}}
\put(316,-3571){\makebox(0,0)[lb]{\smash{{\SetFigFont{12}{14.4}{\familydefault}{\mddefault}{\updefault}{\color[rgb]{0,0,0}Enc. $2$}%
}}}}
\put(4456,-2851){\makebox(0,0)[lb]{\smash{{\SetFigFont{12}{14.4}{\familydefault}{\mddefault}{\updefault}{\color[rgb]{0,0,0}Dec. $1$}%
}}}}
\put(4456,-3571){\makebox(0,0)[lb]{\smash{{\SetFigFont{12}{14.4}{\familydefault}{\mddefault}{\updefault}{\color[rgb]{0,0,0}Dec. $2$}%
}}}}
\put(496,-2356){\makebox(0,0)[lb]{\smash{{\SetFigFont{12}{14.4}{\familydefault}{\mddefault}{\updefault}{\color[rgb]{0,0,0}$W_1$}%
}}}}
\put(496,-4111){\makebox(0,0)[lb]{\smash{{\SetFigFont{12}{14.4}{\familydefault}{\mddefault}{\updefault}{\color[rgb]{0,0,0}$W_2$}%
}}}}
\put(3826,-3706){\makebox(0,0)[lb]{\smash{{\SetFigFont{12}{14.4}{\familydefault}{\mddefault}{\updefault}{\color[rgb]{0,0,0}$\boldsymbol{Y}_2$}%
}}}}
\put(3826,-2716){\makebox(0,0)[lb]{\smash{{\SetFigFont{12}{14.4}{\familydefault}{\mddefault}{\updefault}{\color[rgb]{0,0,0}$\boldsymbol{Y}_1$}%
}}}}
\put(2521,-2716){\makebox(0,0)[lb]{\smash{{\SetFigFont{12}{14.4}{\familydefault}{\mddefault}{\updefault}{\color[rgb]{0,0,0}$\hat{\boldsymbol{U}}(\hat{V}_3)$}%
}}}}
\put(2206,-3526){\makebox(0,0)[lb]{\smash{{\SetFigFont{12}{14.4}{\familydefault}{\mddefault}{\updefault}{\color[rgb]{0,0,0}$\boldsymbol{Y}_3$}%
}}}}
\put(2746,-3526){\makebox(0,0)[lb]{\smash{{\SetFigFont{12}{14.4}{\familydefault}{\mddefault}{\updefault}{\color[rgb]{0,0,0}$\boldsymbol{X}_3(\hat{V}_3)$}%
}}}}
\put(-269,-3751){\makebox(0,0)[lb]{\smash{{\SetFigFont{12}{14.4}{\familydefault}{\mddefault}{\updefault}{\color[rgb]{0,0,0}$\boldsymbol{U}(V_3)$}%
}}}}
\put(1036,-2716){\makebox(0,0)[lb]{\smash{{\SetFigFont{12}{14.4}{\familydefault}{\mddefault}{\updefault}{\color[rgb]{0,0,0}$\boldsymbol{X}_1(W_1)$}%
}}}}
\put(1036,-3706){\makebox(0,0)[lb]{\smash{{\SetFigFont{12}{14.4}{\familydefault}{\mddefault}{\updefault}{\color[rgb]{0,0,0}$\boldsymbol{X}_2(W_2)$}%
}}}}
\put(5356,-2896){\makebox(0,0)[lb]{\smash{{\SetFigFont{12}{14.4}{\familydefault}{\mddefault}{\updefault}{\color[rgb]{0,0,0}$\hat{V}_3$}%
}}}}
\put(5356,-3616){\makebox(0,0)[lb]{\smash{{\SetFigFont{12}{14.4}{\familydefault}{\mddefault}{\updefault}{\color[rgb]{0,0,0}$\hat{V}_3$}%
}}}}
\put(5851,-2896){\makebox(0,0)[lb]{\smash{{\SetFigFont{12}{14.4}{\familydefault}{\mddefault}{\updefault}{\color[rgb]{0,0,0}$\hat{W}_2$}%
}}}}
\put(5851,-3616){\makebox(0,0)[lb]{\smash{{\SetFigFont{12}{14.4}{\familydefault}{\mddefault}{\updefault}{\color[rgb]{0,0,0}$\hat{W}_1$}%
}}}}
\end{picture}%
}
\caption{Relationship among the random variables in FDF-L}
\label{fig:coding}
\end{figure}

\begin{definition}
The jointly strongly $\delta$-typical set $T^n_{[XY]\delta}$ with respect to a distribution $p(x,y)$ on $\mathcal{X}\times\mathcal{Y}$ is the set of sequences $(\boldsymbol{x},\boldsymbol{y}) \in \mathcal{X}^n \times \mathcal{Y}^n$ such that $N(a,b;\boldsymbol{x},\boldsymbol{y})=0$ for $p(a,b)=0$, and $\sum_{a\in\mathcal{X}}\sum_{b\in\mathcal{Y}} \left\vert \frac{1}{n}N(a,b;\boldsymbol{x},\boldsymbol{y}) - p(x,y) \right\vert \leq \delta$,
where $N(a,b;\boldsymbol{x},\boldsymbol{y})$ is the number of occurrences of the pair of symbols $(a,b)$ in the pair of sequences $(\boldsymbol{x},\boldsymbol{y})$, $\delta$ is an arbitrarily small positive real number, and the sequences in $T^n_{[XY]\delta}$ are called strongly jointly $\delta$-typical sequences.
\end{definition}

The relay's estimate $\hat{V}_3$ is the \emph{unique} sequence $\boldsymbol{U}(\hat{V}_3)$ that is jointly strongly $\delta$-typical with its received sequence $\boldsymbol{Y}_3$, i.e., the relay finds $\hat{v}_3$ such that $\Big( \boldsymbol{U}(\hat{v}_3), \boldsymbol{Y}_3 \Big) \in T^n_{[UY_3]|\mathcal{F}|\delta},$ and such that there is no $v' \neq \hat{v}_3$ where $\Big( \boldsymbol{U}(v'), \boldsymbol{Y}_3 \Big) \in T^n_{[UY_3]|\mathcal{F}|\delta}$.

Now, we bound the probability that the relay wrongly decodes $V_3$.
Let $W_1=a_1$ and $W_2=a_2$ be the transmitted messages, and $V_3=b$ the corresponding index for $\boldsymbol{U}$. Let
\vspace{-0.5ex}
\begin{align}
E_0 &= \left\{ (\boldsymbol{U}(b), \boldsymbol{Y}_3) \notin T^n_{[UY_3]|\mathcal{F}|\delta} \right\}\\
E_1 &= \left\{ \exists v_3' \neq b: (\boldsymbol{U}(v_3'), \boldsymbol{Y}_3) \in T^n_{[UY_3]|\mathcal{F}|\delta} \right\}\\
E_2 &=  E_0 \cup E_1 .
\end{align}
$E_2$ is the event that node 3 wrongly decodes $V_3$. In addition, we define the following event:
\vspace{-0.5ex}
\begin{equation}
E_3 = \left\{ (\boldsymbol{X}_1(a_1),\boldsymbol{X}_2(a_2),\boldsymbol{Y}_3) \in T^n_{[X_1X_2Y_3]\delta} \right\}.
\end{equation}
$E_3$ is the event that the received signal $\boldsymbol{Y}_3$ is jointly strongly $\delta$-typical with the users' codewords $\boldsymbol{X}_1(a_1)$ and $\boldsymbol{X}_2(a_2)$.

Since the users' codewords $\boldsymbol{X}_1$ and $\boldsymbol{X}_2$ share the same $\mathbb{G}$ but have independently generated $\boldsymbol{q}_1$ and $\boldsymbol{q}_2$, from Lemma~\ref{lemma:two-linear-codes}, the codewords are independent. Furthermore, from Lemma~\ref{lemma:linear-codes-1}, the codeletters $\big\{X_i[t]:i \in \{1,2\}, t\in \{1,2,\dotsc,n\}\big\}$ are also independent. From \cite[Theorem 6.9]{yeung08}, for a sufficiently large $n$, we have $\Pr\{ E_3 \}  > 1 - \delta$, meaning that $(\boldsymbol{X}_1(a_1),\boldsymbol{X}_2(a_2),\boldsymbol{Y}_3)$ is jointly strongly $\delta$-typical with probability tending to one, by choosing a sufficiently small $\delta >0$.

So, for a sufficiently large $n$, $\Pr \{ E_0^c \}$ equals
\vspace{-0.5ex}
\begin{subequations}
\begin{align}
& \Pr \Big\{ (\boldsymbol{U}(b), \boldsymbol{Y}_3) \in T^n_{[UY_3]|\mathcal{F}|\delta} \Big\} \nonumber\\
&= \Pr \Big\{ E_3^c \Big\}\Pr \Big\{ (\boldsymbol{U}(b), \boldsymbol{Y}_3) \in T^n_{[UY_3]|\mathcal{F}|\delta} \Big| E_3^c \Big\} \nonumber\\
&\quad+ \Pr \Big\{ E_3 \Big\} \Pr \Big\{ (\boldsymbol{U}(b), \boldsymbol{Y}_3) \in T^n_{[UY_3]|\mathcal{F}|\delta} \Big| E_3 \Big\}\\
& > \alpha + (1-\delta)\Pr \Big\{ (\boldsymbol{U}(b), \boldsymbol{Y}_3) \in T^n_{[UY_3]|\mathcal{F}|\delta} \Big| E_3 \Big\}\\
&>  \alpha + (1-\delta)(1-\epsilon),\label{eq:markov-lemma}
\end{align}
\end{subequations}
for some arbitrarily small $\epsilon >0$, 
where $\alpha \triangleq \Pr \{ E_3^c \}\Pr \big\{ \boldsymbol{U}(b), \boldsymbol{Y}_3 \in T^n_{[UY_3]|\mathcal{F}|\delta} \big| E_3^c \big\} \leq  \Pr \{ E_3^c \} < \delta$. Eqn. \eqref{eq:markov-lemma} follows from the Markov Lemma~\cite[page 202 (Lemma 4.1)]{berger77} because $U - (X_1,X_2) - Y_3$ forms a Markov chain.

\begin{remark}
Note that $(\boldsymbol{X}_1,\boldsymbol{X}_2,\boldsymbol{Y}_3)$ being jointly strongly $\delta$-typical does not imply that $(\boldsymbol{U},\boldsymbol{Y}_3)$ is jointly strongly $\delta$-typical.  However, since $U - (X_1, X_2) - Y_3$ forms a Markov chain, invoking the Markov lemma yields that $(\boldsymbol{U}, \boldsymbol{Y}_3)$ is jointly strongly $\delta$-typical with probability tending to one.
%Recall that $U - (X_1,X_2) - Y_3$ forms a Markov chain. So, that $(\boldsymbol{X}_1,\boldsymbol{X}_2,\boldsymbol{Y}_3)$ is jointly strongly $\delta$-typical does not imply that $(\boldsymbol{U},\boldsymbol{Y}_3)$ is jointly strongly $\delta$-typical. But by invoking Markov Lemma, we have shown that $(\boldsymbol{U}, \boldsymbol{Y}_3)$ is jointly strongly $\delta$-typical with probability tending to one.
\end{remark}

It follows that
\vspace{-0.5ex}
\begin{equation}
\Pr\{E_0\} = 1 - \{ E_0^c\} < \delta + \epsilon - \delta\epsilon - \alpha < \epsilon_0,
\end{equation}
for some arbitrarily small $\epsilon_0 > 0$, by choosing a sufficiently small $\delta$.

Now, from Lemma~\ref{lemma:linear-codes-2}, for any $v_3' \neq b$,  $\boldsymbol{U}(v_3')$ and $\boldsymbol{U}(b)$ are independent, and hence $\boldsymbol{U}(v_3')$ and $\boldsymbol{Y}_3$ are also independent. So, we have $\Pr \{ E_1 \}$ equals 
\vspace{-0.5ex}
\begin{subequations}
\begin{align}
& \Pr \Big\{ \exists v_3' \neq b: (\boldsymbol{U}(v_3'), \boldsymbol{Y}_3) \in T^n_{[UY_3]|\mathcal{F}|\delta} \Big\}\nonumber\\
&\leq \sum_{\substack{v_3' \in \{1,2,\dotsc,2^{nR}\} \setminus \{b\}}} \Pr \Big\{ (\boldsymbol{U}(v_3'), \boldsymbol{Y}_3) \in T^n_{[UY_3]|\mathcal{F}|\delta} \Big\} \label{eq:union-bound-2}\\
&= (2^{nR}-1)\Pr \Big\{( \boldsymbol{U}(v_3'), \boldsymbol{Y}_3) \in T^n_{[UY_3]|\mathcal{F}|\delta} \Big\}\\
&\leq  (2^{nR}-1) 2^{-n[I(U;Y_3)-\tau]} \label{eq:jaep2}\\
&< 2^{-n [ I(U;Y_3) - \tau - R]} \leq \epsilon_1,\label{eq:end}
\end{align}
\end{subequations}
for some arbitrarily small $\epsilon_1 > 0$ if $n$ is sufficiently large and if $R < I(U;Y_3) - \tau$, where $\tau \rightarrow 0$ as $|\mathcal{F}|\delta \rightarrow 0$.
Here \eqref{eq:union-bound-2} is by the union bound, and \eqref{eq:jaep2} follows from \cite[Lemma 7.17]{yeung08} as $\boldsymbol{U}(v_3')$ and $\boldsymbol{Y}_3$ are independent.

Hence, if
\vspace{-0.5ex}
\begin{equation}
R < I(U;Y_3), \label{eq:sketch-1}
\end{equation}
for $p(u,y_3)$ defined in \eqref{eq:p-u-y3}, then $\Pr\{ E_2 \}  =\Pr \{ E_0 \cup E_1 \} \leq \Pr\{E_0\} + \Pr\{E_1\} < \epsilon_0 + \epsilon_1$, where $\epsilon_0 + \epsilon_1$ can be made arbitrarily small, i.e., the relay can \emph{reliably} decode $\hat{V}_3=b$.

\noindent{\bf Downlink:}\\
\indent Assuming that the relay has correctly  decoded $\boldsymbol{U}(V_3)$, it re-encodes and broadcasts the index $V_3 \in \{1,2,\dotsc,2^{nR}\}$ to the users in $n$ downlink channel uses. For $n$ sufficiently large, the users can reliably decode $V_3$ if~\cite[p. 567 (Theorem 15.6.3)]{coverthomas06}
\vspace{-0.5ex}
\begin{equation}
R < I(X_3;Y_1) \quad\text{and}\quad R < I(X_3;Y_2), \label{eq:sketch-2}
\end{equation}
for some $p(x_3) p^*(y_1,y_2|x_3)$. Note that linear codes are not required on the downlink.

Assuming node 1 correctly decodes the relay's message $V_3$, knowing its own message $W_1$, it can perform $\boldsymbol{U}(V_3) \oplus (- \boldsymbol{X}_1(W_1))$ to get $W_2$, where $(-\boldsymbol{X}_1)$ is the element-wise additive inverse of $\boldsymbol{X}_1$. Node 2 decodes $W_1$ using a similar method.
Combining \eqref{eq:sketch-1} and \eqref{eq:sketch-2}, we have Theorem~\ref{theorem:achievability}. \hfill $\blacksquare$

%In this coding strategy, the relay, instead of decoding both $W_1$ and $W_2$, decodes the function $f_{1,2}(W_1,W_2)=V_3$. The function must be defined with the constraints that $W_2$ can be recovered from $(V_3,W_1)$ and that $W_1$ can be recovered from $(V_3,W_2)$.

\subsection{Functional-Decode-Forward with Systematic Computation Codes (FDF-S)}
An achievable rate region for the discrete memoryless TWRC using FDF can also be obtained by using systematic computation codes~\cite{nazergastpar07} (instead of linear codes) on the uplink. Similar to FDF-L, the relay computes a function of the users' codewords (the function $\boldsymbol{X}_1(W_1) \oplus \boldsymbol{X}_2(W_2)$ can again be chosen) and broadcasts this function back to the users. However, on the uplink, using systematic computation codes, the users first send uncoded transmissions to the relay, followed by a refinement stage in which the users send linear-coded transmissions. We can show that the rate region in the following theorem is achievable for the TWRC.

\begin{theorem}\label{theorem:compute-forward}
Consider a TWRC where $|\mathcal{X}_1| = |\mathcal{X}_2| = |\mathcal{F}|$, for some finite field $\mathcal{F}$. Rename the elements in $\mathcal{X}_1$ and $\mathcal{X}_2$ so that $\mathcal{X}_1 = \mathcal{X}_2 = \mathcal{F}$. The rate pair $(R_1,R_2)$ is achievable if
%\vspace{-0.5ex}
\begin{align}
R_1 &\leq \left[ \frac{C_\text{MAC}H(W_1)}{C_\text{MAC} + 2H(X_1 \oplus X_2|Y_3)}, I(X_3;Y_1),  I(X_3;Y_2) \right] \nonumber \\
R_2 &\leq \left[ \frac{C_\text{MAC}H(W_2)}{C_\text{MAC} + 2H(X_1 \oplus X_2|Y_3)}, I(X_3;Y_1),  I(X_3;Y_2) \right], \nonumber
\end{align}
for some joint distributions of the form $p^*(y_3|x_1,x_2)p(x_1|w_1)p(x_2|w_2)p(w_1,w_2)$ and $p(x_3)p^*(y_1,y_2|x_3)$. Here, $C_\text{MAC}$ is the maximum sum-rate of the multiple-access channel $p^*(y_3|x_1,x_2)$.
\end{theorem}

\begin{remark}
The above result is also valid even for TWRCs where $\mathcal{X}_1$ and $\mathcal{X}_2$ are not finite fields. See Remark~\ref{remark:non-finite-field}.
\end{remark}

The above rate region is obtained using the results in \cite[Theorem 2]{nazergastpar07} (by setting $V=X_1\oplus X_2$) and \cite[p. 567 (Theorem 15.6.3)]{coverthomas06}. The additional factors $H(W_1)$ and $H(W_2)$ in the above equations compared to \cite[Eqn. (23)]{nazergastpar07} convert computation rates to rates in bits/channel use considered in this paper. The proof is omitted because of space constraints.

\section{Comparison of Coding Strategies}
 \label{sec:comparison}

In this section, we show that the maximum sum rate obtained by FDF-L can be simultaneously higher than those achievable by FDF-S, and by two existing coding strategies: CDF and CF.

\subsection{Existing Coding Strategies}

\subsubsection{Complete-Decode-Forward (CDF)}
Using CDF, the relay completely decodes the messages $W_1$ and $W_2$ sent by users 1 and 2 respectively.  It then encodes and broadcasts a function of the messages to the users such that each user can recover the message sent by the other user. The overall achievable rate region is thus limited by two sets of constraints, i.e., the multiple-access constraints~\cite{ahlswede71,liao72} on the uplink and the broadcast constraints~\cite[Theorem 2.5]{oechteringschnurr08} on the downlink, and is given in the following theorem.

\begin{theorem}\label{theorem:cdf}[see~\cite{knopp06,kimmitrantarokh08,namchung08}]
Consider a TWRC. The rate pair $(R_1,R_2)$ is achievable using CDF if
%\vspace{-0.3ex}
\begin{align}
R_1 &\leq \min \Big\{ I(X_1;Y_3|X_2) , I(X_3;Y_2) \Big\}\\
R_2 &\leq \min \Big\{ I(X_2;Y_3|X_1) , I(X_3;Y_1) \Big\}\\
R_1 + R_2 &\leq I(X_1,X_2;Y_3), \label{eq:complete-decode-forward-sum-rate}
\end{align}
for some joint distributions of the form
$p(x_1)p(x_2)p^*(y_3|x_1,x_2)$ and
$p(x_3)p^*(y_1,y_2|x_3)$.
\end{theorem}

\subsubsection{Compress-Forward (CF)}
Using this strategy the relay quantizes its received signal $Y_3$ to $\hat{Y}_3$, encodes and broadcasts $\hat{Y}_3$ to the users. Assuming that both users can correctly decode $\hat{Y}_3$, a virtual channel $X_1 \rightarrow \hat{Y}_3$ is created from user 1 to user 2 via the relay. Similarly, a virtual channel $X_2 \rightarrow \hat{Y}_3$ is created from user 2 to user 1 via the relay. The achievable rate region using CF on the TWRC is given in the following theorem.

\begin{theorem}\label{theorem:cf}[see~\cite{schnurroechtering07}]
Consider a TWRC. The rate pair $(R_1,R_2)$ is achievable using CF if $R_1 \leq I(X_1;\hat{Y}_3|X_2,T)$ and $R_2 \leq I(X_2;\hat{Y}_3|X_1,T)$,
under the constraints $I(Y_3;\hat{Y}_3|X_1,T) < I(X_3;Y_1)$ and $I(Y_3;\hat{Y}_3|X_2,T) < I(X_3;Y_2)$, for some joint distributions of the form $p(t)p(x_1|t)p(x_2|t)p^*(y_3|x_1,x_2)p(\hat{y}_3|y_3)$ and $p(x_3)p^*(y_1,y_2|x_3)$, where $T \in \mathcal{T}$, $|\mathcal{T}| \leq 4$, $\hat{Y}_3 \in \hat{\mathcal{Y}}_3$, and  $|\hat{\mathcal{Y}}_3| \leq |\mathcal{Y}_3|+3$.
\end{theorem}

\subsection{Numerical Calculations}

Now, we compare these four coding strategies on a TWRC.

\subsubsection{Channel}

We consider the following TWRC:
\begin{itemize}
\item $\mathcal{X}_1 = \mathcal{X}_2 = \mathcal{X}_3 = \mathcal{Y}_1 = \mathcal{Y}_2 = \{0,1\}$, $\mathcal{Y}_3 = \{a,b,c,d\}$.
\item $p^*(y_3|x_1,x_2)$ is given by the following transition matrix:

\vspace{0.3ex}
%\begin{equation*}\label{eq:example-channel}
\begin{small}
\begin{tabular}{c | c c c c}
$p^*(y_3|x_1,x_2)$  & $y_3=a$ & $y_3=b$ & $y_3=c$ & $y_3=d$ \\
\hline
$x_1=0, x_2=0$ & 0.1 & 0.2 & 0.3 & 0.4\\
$x_1=0, x_2=1$ & 0.3 & 0.4 & 0.1 & 0.2\\
$x_1=1, x_2=0$ & 0.4 & 0.3 & 0.2 & 0.1\\
$x_1=1, x_2=1$ & 0.2 & 0.1 & 0.4 & 0.3
\end{tabular}
\end{small}.
\vspace{0.2ex}
%\end{equation*}

Each entry in the lower right matrix denotes the conditional probability $p^*(y_3|x_1,x_2)$ that $y_3$ is received when $(x_1,x_2)$ are sent. Note that $p^*(y_3|x_1,x_2)$ cannot be written as a noisy function of $x_1\oplus x_2$. % cascaded with a symmetrical channel $p(y_3|v)$. So linear computation coding~\cite[Section IV]{nazergastpar07} cannot be used for the uplink.
\item $p^*(y_1,y_2|x_3) = \text{BSC}_{Y_1|X_3} (0.3)\text{BSC}_{Y_2|X_3} (0.3)$, where
\vspace{-1ex}
\begin{equation}
\text{BSC}_{Y|X}(\rho) \triangleq p(y|x)= \begin{cases}
1-\rho, & \text{if } y=x\\
\rho, &\text{otherwise}
\end{cases}.
\end{equation}
The downlink from the relay to each user is a binary-symmetric channel with cross-over probability $\rho=0.3$.
\end{itemize}

\subsubsection{Achievable Sum Rates}
The maximum sum rates (i.e., $R_1+R_2$) achievable by the different coding strategies are%: (i) FDF-L: 0.2374, (ii) FDF-S: 0.1602, (iii) CDF: 0.1536, (iv) CF (an upper bound): $R_1+R_2 \leq 0.2374 - \zeta$, for some $\zeta > 0$.
\begin{itemize}
\item FDF-L: $R_1+R_2=0.2374$.
\item FDF-S: $R_1+R_2=0.1602$.
\item CDF: $R_1+R_2=0.1536$.
\item CF (an upper bound on the maximum sum rate): $R_1+R_2 \leq 0.2374 - \zeta$, for some $\zeta > 0$.
\end{itemize}
Clearly, FDF-L outperforms the other coding strategies on this TWRC.

\section{Conclusion} \label{sec:conclusion}

We have proposed a functional-decode-forward coding strategy with linear codes (FDF-L) for the general discrete memoryless two-way relay channel (TWRC) and obtained a new achievable rate region. We showed that using random linear codes for the users, the relay can reliably decode a function of the users' codewords even when the channel does not perform the desired function. The function, when broadcast back to the users, allows each user to decode the other user's message.

Noting that functional decoding on the uplink of the discrete memoryless TWRC is also possible using systematic computation codes, we obtained another achievable region for the TWRC using functional-decode-forward with systematic computation codes (FDF-S).

With an example, we numerically showed that FDF-L is capable of achieving strictly higher sum rates compared to FDF-S and two existing coding strategies, namely, complete-decode-forward and compress-forward.

However, using FDF-L or FDF-S, if the cardinalities of the user's input alphabets $|\mathcal{X}_i|$ are both not equal to that of any finite field, only subsets of $\mathcal{X}_i$ are utilized for transmission. Furthermore, since linear codes are used for FDF-L, the distributions of the users' transmitted signals $p(x_i)$ are constrained to be uniform, which is not always optimal for the channel. %In addition, using FDF-L or FDF-S, the cardinality of the channel input alphabets of the users is limited to that of a finite field, i.e., $\ell^z$ for some prime number $\ell$ and some positive integer $z$.

This paper nonetheless provides coding schemes for the relay to decode a function of the users' messages without having to decode the messages individually on the general discrete memoryless TWRC (which may not be additive). This strategy can be useful in multiterminal networks where different destination nodes have knowledge of some source messages and want to decode the messages of other sources.

%\bibliography{../bib}

% Generated by IEEEtran.bst, version: 1.13 (2008/09/30)

\end{document}